\documentclass[aps,prx,superscriptaddress,twocolumn,twoside,nofootinbib,floatfix,a4paper]{revtex4-2}

\usepackage{graphicx}
\usepackage{times}
\usepackage{physics}
\usepackage{amssymb,amsthm}
\usepackage{textcomp}
\usepackage{mathtools}
\usepackage[utf8]{inputenc}
\usepackage[polish,english]{babel}
\usepackage[shortlabels]{enumitem}
\usepackage[T1]{fontenc}
\usepackage{xcolor}
\usepackage{dcolumn}
\usepackage{bm}
\usepackage{verbatim}
\usepackage{amsmath}
\usepackage{dsfont,bbm}
\usepackage{pifont}
\usepackage{slashed}
\usepackage{orcidlink}

\theoremstyle{plain}

\newtheorem{theorem}{Theorem}

\newtheorem{definition}{Definition}

\usepackage{hyperref}
\hypersetup{
    colorlinks,
    citecolor=blue,
    filecolor=black,
    linkcolor=blue,
    urlcolor=blue
}
\usepackage{cleveref}

\usepackage[normalem]{ulem}
\newcommand{\stkout}[1]{\ifmmode\text{\sout{\ensuremath{#1}}}\else\sout{#1}\fi}

\let\existstemp\exists
\let\foralltemp\forall
\renewcommand*{\exists}{\existstemp\mkern2mu}
\renewcommand*{\forall}{\foralltemp\mkern2mu}

\begin{document}

\title{Device-independent certification of tripartite quantum networks with bilocal Bell inequalities}

\author{Patryk Michalski\orcidlink{0009-0009-0305-7356}}
\email{pmichalski@cft.edu.pl}
\affiliation{Center for Quantum-Enabled Computing, Center for Theoretical Physics, Polish Academy of Sciences, al. Lotnik\'{o}w 32/46, 02-668 Warsaw, Poland}
\affiliation{Institute of Theoretical Physics, University of Warsaw, Pasteura 5, 02-093 Warsaw, Poland}
\author{Arturo Konderak\orcidlink{0000-0002-4546-2626}}
\email{akonderak@cft.edu.pl}
\affiliation{Center for Quantum-Enabled Computing, Center for Theoretical Physics, Polish Academy of Sciences, al. Lotnik\'{o}w 32/46, 02-668 Warsaw, Poland}
\author{Remigiusz Augusiak\orcidlink{0000-0003-1154-6132}}
\email{augusiak@cft.edu.pl}
\affiliation{Center for Quantum-Enabled Computing, Center for Theoretical Physics, Polish Academy of Sciences, al. Lotnik\'{o}w 32/46, 02-668 Warsaw, Poland}

\begin{abstract}
While quantum networks have been extensively studied as a natural extension of the standard Bell scenario with a richer correlation structure, general constructions of nonlinear Bell inequalities with self-testing properties are still largely lacking. In this work, we present a general method for constructing such inequalities in the simplest network scenario, in which two independent sources distribute bipartite quantum states to three spatially separated observers. These inequalities allow for arbitrary numbers of binary measurements and are maximally violated by maximally entangled states of the corresponding local dimensions together with sets of pairwise anticommuting Clifford observables. Importantly, their maximal quantum values can be determined analytically, which makes them particularly promising for device-independent applications. In particular, we prove that these Bell inequalities can be used to device-independently certify the underlying quantum network, including both the quantum states produced by the sources and the observables measured by all parties. To the best of our knowledge, this is the first self-testing result for quantum networks
that relies solely on the maximal violation of a nonlocality witness.
\end{abstract}

\maketitle

\textit{Introduction.}---Quantum networks have emerged as a key subject in the study of Bell nonlocality, since they offer a natural setting for modeling realistic scenarios involving multiple spatially separated observers. While standard Bell scenarios assume a single common source, network configurations allow correlations to be distributed by multiple independent sources~\cite{Tavakoli2022}. Consequently, nonlocal correlations in networks exhibit a richer structure, as source independence fundamentally alters the geometry of both classical and quantum sets of correlations. Even in the simplest scenario involving two independent sources and three parties---the so-called bilocal scenario---the source independence assumption significantly complicates the derivation and analysis of optimal nonlocality witnesses, as it renders the set of local correlations non-convex. This naturally motivates the introduction of nonlinear Bell inequalities~\cite{Fritz2012}. 

The first paradigmatic example of such a nonlinear Bell inequality is the Branciard--Rosset--Gisin--Pironio (BRGP) inequality for the bilocal scenario~\cite{Branciard2010,Branciard2012}, which was subsequently extended to broader network topologies and measurement structures. These include general network architectures~\cite{Rosset2016,Tavakoli2016,Hsu2021,Luo2024}, star and chain networks~\cite{Tavakoli2014,Mukherjee2015,Luo2024}, and setups tailored to the elegant joint measurement~\cite{Gisin2019,Tavakoli2021}. Simultaneously, several constructions of linear Bell inequalities for quantum networks have also been proposed~\cite{Luo2022,Luo2024}. However, most existing constructions are restricted either to scenarios with only two measurements for the external parties or to qubit systems. In fact, Bell inequalities for quantum networks involving more than two measurement settings per party, as well as those tailored to quantum systems with higher local dimensions, remain largely unexplored. Moreover, Bell inequalities capable of self-testing quantum networks are still lacking, even in the simplest bilocal setup.

\begin{figure}[ht]
    \centering
    \includegraphics[width=\linewidth]{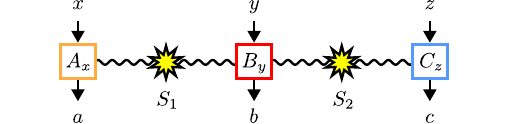}
    \caption{Schematic representation of the bilocal (or three-on-a-line) network with two independent sources, $S_1$ and $S_2$, distributing bipartite correlations among three spatially separated observers, Alice, Bob, and Charlie. The observers choose their measurements freely and independently, labeled by $x$, $y$, and $z$, and obtain outcomes $a$, $b$, and $c$, respectively.}
    \label{fig:network}
\end{figure}

In this work, we address the aforementioned issues by developing a general framework for nonlinear Bell inequalities tailored to tripartite quantum networks with two independent sources (cf. Fig.~\ref{fig:network}). These inequalities involve arbitrary numbers of binary measurements for the external and central parties. Crucially, they exhibit maximal violation for maximally entangled states of suitable local dimensions emitted by the sources, together with sets of pairwise anticommuting Clifford observables measured by the external parties. 
Beyond their importance in foundational physics~\cite{Araki2003,Derezinski2023,Szalay2021}, such observables are highly relevant in topological quantum computation~\cite{Lahtinen2017,Karzig2017}, uncertainty relations \cite{WehnerWinter2008}, and quantum information protocols involving steering inequalities with unbounded quantum violations \cite{Rutkowski2015}.

Our construction is based on a broad family of real row-orthogonal and column-normalized (ROCN) matrices~\cite{Michalski2025}. Importantly, the algebraic properties of ROCN matrices enable the analytical derivation of the maximal quantum values for the resulting nonlinear Bell inequalities---a feature essential for device-independent applications. We show that the maximal quantum value coincides exactly with the number of measurements performed by the central party. Furthermore, we characterize the corresponding bilocal bound and establish a necessary and sufficient condition for a nontrivial quantum violation.

Perhaps the most remarkable feature of the proposed inequalities is that, under a suitable condition on the underlying ROCN matrices, they enable an almost complete (up to certain equivalences) characterization of the entire network. Specifically, they allow one to certify both the distributed quantum states and the measurements performed by all parties. Indeed, maximal violations of these nonlinear inequalities, as well as of the corresponding linear counterparts involving the same combinations of expectation values, enable full network self-testing. 

Although self-testing in this class of quantum networks has been previously explored~\cite{Supic2023,Sarkar2026}, the self-testing statement presented here is arguably the first general result of this kind based on the violation of a single nonlinear Bell inequality and its linear counterpart. Therefore, our approach is conceptually cleaner than the existing protocols and more practical for device-independent certification, as it relies essentially on a single experimentally accessible witness.

The results of this work therefore provide a unified framework in which ROCN bilocal inequalities serve both as witnesses of nonbilocality and as certification tools for tripartite quantum networks.

\textit{Preliminaries.}---In the bilocal network scenario, two external parties, Alice and Charlie, each share a source with a central party, Bob, as illustrated in Fig.~\ref{fig:network}. The two independent sources $S_1$ and $S_2$ emit quantum states $\rho_\mathrm{AB_1}$ and $\rho_\mathrm{B_2 C}$ acting on finite-dimensional Hilbert spaces $\mathcal{H}_\mathrm{AB_1} = \mathcal{H}_\mathrm{A} \otimes \mathcal{H}_\mathrm{B_1}$ and $\mathcal{H}_\mathrm{B_2 C} = \mathcal{H}_\mathrm{B_2} \otimes \mathcal{H}_\mathrm{C}$, respectively. Consequently, Bob’s total Hilbert space decomposes as the tensor product $\mathcal{H}_\mathrm{B} = \mathcal{H}_\mathrm{B_1} \otimes \mathcal{H}_\mathrm{B_2}$. 

On their respective subsystems, Alice, Bob, and Charlie perform two-outcome measurements that are freely chosen from three independent sets. By Naimark's dilation theorem~\cite{Holevo2012}, it is sufficient to restrict attention to pure states and projective measurements. We can therefore represent 
these measurements by $\pm1$-eigenvalued observables, denoted by $A_x$, $B_y$ and $C_z$ for Alice, Bob and Charlie, where $x, z \in \{1,\ldots,m\}$ and $y \in \{1,\ldots,n\}$. Formally, $A_x$, $B_y$ and $C_z$ are Hermitian operators acting on $\mathcal{H}_\mathrm{A}$, $\mathcal{H}_\mathrm{B}$ and $\mathcal{H}_\mathrm{C}$, respectively, which satisfy $A_x^2 = B_y^2 = C_z^2 = \mathds{1}$. Here, with a slight abuse of notation, $\mathds{1}$ denotes the identity on the corresponding local Hilbert space. The outcomes obtained by Alice, Bob and Charlie are denoted by $a,b,c \in \{-1,1\}$.

The correlations generated when the parties repeat their measurements over many rounds are described by a family of conditional probability distributions $\{p(a,b,c|x,y,z)\}$, where $p(a,b,c|x,y,z)$ denotes the probability that Alice, Bob, and Charlie obtain outcomes $a$, $b$ and $c$ given the measurement settings $x$, $y$ and $z$, respectively. The corresponding joint expectation values, also referred to as correlators, are defined as 
\begin{equation}\label{correlators}
    \langle A_x B_y C_z \rangle = \sum_{a,b,c=\pm 1} abc \,p(a,b,c|x,y,z).
\end{equation}
Explicitly, for a quantum realization, these correlators are given by Born's rule: 
\begin{equation}\label{eq:quantum_corr}
    \langle A_x B_y C_z \rangle=\Tr(\rho_{\mathrm{AB_1}}\otimes \rho_{\mathrm{B_2C}} \,A_x\otimes B_y\otimes C_z).
\end{equation}

In local hidden variable models for network correlations, the sources $S_1$ and $S_2$ are associated with independent variables $\lambda_1$ and $\lambda_2$, respectively. Denoting by $\mu_i(\lambda_i)$ the probability density of $\lambda_i$, the correlators admit the decomposition
\begin{align}\label{local}
    \langle A_x B_y C_z \rangle =&\! \int \mathrm{d}\lambda_1\, \mu_1(\lambda_1) \int \mathrm{d}\lambda_2\, \mu_2(\lambda_2) \nonumber\\
    &\times \langle A_x \rangle_{\lambda_1} \langle B_y \rangle_{\lambda_1,\lambda_2} \langle C_z \rangle_{\lambda_2},
\end{align}
where $-1\leqslant \langle A_x \rangle_{\lambda_1}, \langle B_y \rangle_{\lambda_1,\lambda_2}, \langle C_z \rangle_{\lambda_2} \leqslant 1$ for any $\lambda_1, \lambda_2$ and any measurement choices $x,y,z$. 

In the standard tripartite Bell scenario with a single hidden variable $\lambda$, the set of all local correlations forms a convex polytope. Its vertices correspond to local deterministic correlations, which satisfy $\langle A_x B_y C_z\rangle = \langle A_x \rangle \langle B_y \rangle \langle C_z \rangle$ for fixed $\lambda$, with the individual expectation values $\langle A_x\rangle$, $\langle B_y\rangle$ and $\langle C_z \rangle$ taking values in $\{-1,1\}$ for any $x,y,z$. When source independence is imposed, the set of correlations attainable in a network acquires a more complex structure. The resulting network-local (or bilocal) set is contained within the conventional local polytope and shares the same deterministic vertices, but it is no longer convex~\cite{Branciard2010}.

Quantum correlations~\eqref{eq:quantum_corr} that cannot be expressed through any bilocal model are termed \emph{nonbilocal}. As in the standard Bell scenario, such correlations can be detected using appropriately chosen inequalities. However, because the network-local set is nonconvex, the optimal witnesses of network nonlocality are necessarily non-linear~\cite{Fritz2012}. Nevertheless, deriving such nonlinear inequalities remains a notoriously challenging task~\cite{Tavakoli2022}.

We consider a family of nonlinear Bell inequalities constructed from linear combinations of correlators (\ref{correlators}). For each input $y$, we define
\begin{equation}\label{eq:I_y}
    I_y \vcentcolon= \sum_{x = 1}^m \sum_{z = 1}^m h_{xy} h_{zy} \langle A_x B_y C_z \rangle,
\end{equation}
where $h$ is a real-valued $m\times n$ matrix. These quantities are then combined into the nonlinear expression
\begin{equation}\label{eq:S_inequality}
    \mathcal{S}_h \vcentcolon= \sum_{y = 1}^n \sqrt{\left| I_y \right|} \leqslant \beta_\mathrm{BL},
\end{equation}
where the constant $\beta_\mathrm{BL}$ denotes the maximal value of $\mathcal{S}_h$ achievable by correlations that admit a bilocal decomposition~\eqref{local}, and we refer to it as the \emph{bilocal bound}.

Notably, this description generalizes the pioneering example of the 22-type BRGP inequality~\cite{Branciard2012}, which corresponds to the simplest bilocal scenario in which each party chooses between two possible two-outcome observables ($m=n=2$). In this case, the combinations $I_y^\mathrm{BRGP}$ are defined as
\begin{align}
    I_1^\mathrm{BRGP} \vcentcolon={} \frac{1}{\sqrt{2}} (& \langle A_1 B_1 C_1 \rangle - \langle A_1 B_1 C_2 \rangle  \nonumber\\
    &- \langle A_2 B_1 C_1 \rangle + \langle A_2 B_1 C_2 \rangle),\nonumber\\
    I_2^\mathrm{BRGP} \vcentcolon={} \frac{1}{\sqrt{2}} (& \langle A_1 B_1 C_1 \rangle + \langle A_1 B_1 C_2 \rangle \nonumber\\
    &+ \langle A_2 B_1 C_1 \rangle + \langle A_2 B_1 C_2 \rangle),
\end{align}
leading to the nonlinear inequality
\begin{equation}\label{eq:BRGP}
    \mathcal{S}_\mathrm{BRGP} \vcentcolon= \sqrt{\left|I_1^\mathrm{BRGP}\right|} + \sqrt{\left|I_2^\mathrm{BRGP}\right|} \leqslant \sqrt{2}.
\end{equation}
The maximal quantum value of $\mathcal{S}_\mathrm{BRGP}$ is $\beta_\mathrm{Q}^{\mathrm{BRGP}}=2$, achieved when each source emits a maximally entangled two-qubit state $\ket{\Phi_2}=(|00\rangle+|11\rangle)/\sqrt{2}$, and the local observables are $A_1=C_1=(X+Z)/\sqrt{2}$, $A_2=C_2=(X-Z)/\sqrt{2}$ for Alice and Charlie, and $B_1=X \otimes X$, $B_2=Z \otimes Z$ for Bob. Here, $X$, $Y$ and $Z$ denote the Pauli matrices. 

Note that the observables of Alice and Charlie anticommute, i.e., $\{A_1,A_2\}=\{C_1,C_2\}=0$, whereas those of Bob commute, $[B_1,B_2]=0$. Furthermore, the BRGP inequality is often studied in the context of entanglement swapping~\cite{Branciard2012}, where Bob performs a single four-outcome measurement, thereby enabling Alice and Charlie to share an entangled state. In fact, it has been demonstrated~\cite{Branciard2012,Andreoli2017} that the nonbilocal correlations generated via entanglement swapping can always be mapped to correlations that are attainable in a scenario featuring two binary measurements for Bob. 

\textit{ROCN Bell inequalities for bilocal networks.}---We now proceed to define the family of ROCN bilocal inequalities, which extend the ROCN Bell inequalities introduced in Ref.~\cite{Michalski2025} to the network scenario. This extension enables the derivation of a wide class of Bell inequalities for detecting nonbilocality and self-testing the quantum network. Central to this construction is the notion of an ROCN matrix, whose definition we recall below for completeness.

\begin{definition}[ROCN matrix]\label{def:ROCN_matrix}
    Let $h$ be a real $m\times n$ matrix with $m\leqslant n$. We say that $h$ is a \emph{row-orthogonal and column-normalized (ROCN) matrix} if its rows are pairwise orthogonal and nonzero, and its columns are normalized.
\end{definition}

Any ROCN matrix $h\in\mathds{R}^{m \times n}$ can be used to construct a nonlinear Bell inequality of the form introduced in Eq.~\eqref{eq:S_inequality} by treating its entries as coefficients in the combinations $I_y$. The resulting inequality reads
\begin{equation}\label{eq:S_h_inequality}
    \mathcal{S}_h \vcentcolon= \sum_{y = 1}^n \sqrt{\left| \sum_{x = 1}^m \sum_{z = 1}^m h_{xy} h_{zy} \langle A_x B_y C_z \rangle \right|} \leqslant \beta^h_\mathrm{BL},
\end{equation}
and we refer to it as the \emph{ROCN bilocal inequality} associated with the matrix $h$. The maximal quantum value of any ROCN bilocal inequality can be analytically determined via the following theorem.

\begin{theorem}\label{th:quantum_bound}
    For any ROCN matrix $h\in\mathds{R}^{m \times n}$, the maximal quantum value of the corresponding ROCN bilocal inequality~\eqref{eq:S_h_inequality} is $\beta^h_\mathrm{Q} = n$ and equals the number of Bob's measurements.
\end{theorem}

The proof of this theorem proceeds in two main steps. The first, deferred to Appendix~\ref{sec:quantum_bound}, shows–--using a suitable sum-of-squares decomposition–--that for any ROCN matrix $h$, the bilocal expression $\mathcal{S}_h$ is bounded from above by $n$, which is the number of Bob's observables. The second step is to provide a quantum strategy that saturates this bound. This strategy relies on two maximally entangled states of appropriate local dimensions emitted by the sources, along with two sets of pairwise anticommuting observables measured by Alice and Charlie. Given its relevance to self-testing, we explicitly specify this quantum strategy below for arbitrary $m$ and $n$.

The explicit form of Alice’s and Charlie’s observables is derived via the Jordan--Wigner representation theorem~\cite{Jordan1928,Samoilenko1991}, which provides an explicit realization of families of pairwise anticommuting operators on multi-qubit Hilbert spaces. The resulting collection of observables constitutes a representation of the generators of a Clifford algebra. Owing to their close connection with the operator formalism of quantum field theory~\cite{Greiner1996}, such operators are often referred to as Majorana fermions. For completeness, the explicit form of the Jordan--Wigner representation is recalled in Appendix~\ref{sec:jordan_wigner}.

We now define the optimal quantum strategy that attains $\beta_Q^h=n$. Alice, Bob, and Charlie have access to the Hilbert spaces $\tilde{\mathcal{H}}_{\mathrm{A}}$, $\tilde{\mathcal{H}}_{\mathrm{B}} = \tilde{\mathcal{H}}_{\mathrm{B}_1} \otimes \tilde{\mathcal{H}}_{\mathrm{B}_2}$, and $\tilde{\mathcal{H}}_{\mathrm{C}}$, respectively, with $\tilde{\mathcal{H}}_{\mathrm{A}} = \tilde{\mathcal{H}}_{\mathrm{B}_1} = \tilde{\mathcal{H}}_{\mathrm{B}_2} = \tilde{\mathcal{H}}_{\mathrm{C}} = (\mathbb{C}^{2})^{\otimes r}$ and $r = \lfloor m/2 \rfloor$. Each of the two independent sources, $S_1$ and $S_2$, emits the maximally entangled state
\begin{equation}\label{eq:me_state}
    \ket{ \Phi_d}=\frac{1}{\sqrt d}\sum_{\boldsymbol{k}\in\{0,1\}^r}\ket{\boldsymbol{k}} \otimes\ket{\boldsymbol{k}} \in (\mathbb{C}^{2})^{\otimes r} \otimes (\mathbb{C}^{2})^{\otimes r},
\end{equation}
where $d=2^r$, $\boldsymbol{k}=(k_1,\dots,k_r)$, and $\ket{\boldsymbol{k}}=\ket{k_1} \otimes \dots \otimes \ket{k_r}$ denotes the computational basis. The two copies of the state~\eqref{eq:me_state} are distributed across the bipartitions $(\tilde{ \mathcal H}_{\mathrm A}, \tilde{\mathcal H}_{\mathrm B_1})$ and $(\tilde{\mathcal H}_{\mathrm B_2},\tilde{\mathcal H}_{\mathrm C})$. Alice and Charlie have access to the observables $\{\tilde A_x\}_{x=1}^m$ and $\{\tilde C_z\}_{z=1}^m$ respectively, which satisfy the canonical anticommutation relations, $\{\tilde A_x,\tilde A_z\}=\{\tilde C_x,\tilde C_z\}=2\delta_{xz}\mathds 1$. Their explicit form is given by Eqs.~\eqref{eq:U} and~\eqref{eq:last_operator_JW} in Appendix~\ref{sec:jordan_wigner}, with $\mathcal H'=\mathbb C$ and $\Gamma_m'=1$. 
Finally, Bob’s observables are defined as a bilinear combination of the transposed operators of Alice and Charlie:
\begin{equation}\label{eq:bob_observables}
       \tilde B_y = \sum_{x=1}^m \sum_{z=1}^m h_{xy} h_{zy} \, \tilde{A}_x^\mathsf{T} \otimes \tilde{C}_z^\mathsf{T},
\end{equation}
with $\mathsf{T}$ denoting the transposition with respect to the canonical basis. One can easily verify that, with this definition, $\tilde B_y^2=\mathds 1$. Altogether, we call the above strategy the \emph{reference strategy} for achieving the maximal quantum bound.

\begin{definition}\label{def:reference_strategy}
    For any ROCN matrix $h\in\mathds{R}^{m \times n}$ and its associated inequality \eqref{eq:S_h_inequality}, the \emph{reference strategy} is defined as a configuration wherein Alice and Charlie each share a copy of the maximally entangled state~\eqref{eq:me_state} with Bob, and measure the anticommuting observables $\{\tilde A_x\}$ and $\{\tilde C_z\}$ detailed in Appendix~\ref{sec:jordan_wigner}, while Bob's observables $\{\tilde B_y\}$ are formed as bilinear combinations of the transposed operators of Alice and Charlie, according to Eq.~\eqref{eq:bob_observables}.
\end{definition}

Notice that $\ket{\Phi_d}$ in Eq.~\eqref{eq:me_state} factorizes as the tensor product of $r$ copies of the two-qubit maximally entangled state $\ket{\Phi_2}$. Furthermore, Bob’s measurement is separable with respect to his two local Hilbert subspaces. This separability significantly simplifies its experimental implementation~\cite{Andreoli2017}.

The explicit form of the quantum bound for bilocal ROCN Bell inequalities is essential for proving the self-testing statement. However, detecting nonbilocality also requires determining the corresponding bilocal bound. As shown in Appendix~\ref{appendix:bilocal_bound}, this bilocal bound coincides with the classical bound of the standard bipartite ROCN Bell inequality discussed in~\cite{Michalski2025}. Consequently, this local bound does not admit a closed-form expression, and the necessary and sufficient conditions for a nontrivial quantum violation are further analyzed in Appendix~\ref{appendix:bilocal_bound}.

\textit{Self-testing from ROCN bilocal inequalities.}---Beyond merely witnessing nonbilocality, we are interested in exploring whether the observed violations of our inequalities can also certify the underlying quantum states and measurements. More concretely, we ask whether the maximal violation of the ROCN bilocal inequality~\eqref{eq:S_inequality} uniquely determines---up to the standard equivalences---the reference strategy introduced in Def.~\ref{def:reference_strategy}. Such a certification is known as self-testing \cite{Mayers2004,Supic2020}.

Our main result demonstrates that the proposed inequalities indeed enable self-testing of the reference quantum strategy, provided that the nonlinear inequality is combined with the following linear ROCN Bell expression:
\begin{equation}\label{eq:rocn_bell_inequality}
	\mathcal{L}_h = \sum_{y=1}^n I_y = \sum_{y = 1}^n \sum_{x = 1}^m \sum_{z = 1}^m h_{xy} h_{zy} \langle A_x B_y C_z \rangle.
\end{equation}
This linear counterpart is composed of the same terms $I_{y}$ as the original nonlinear Bell inequality. Its maximal  quantum value also amounts to $\beta^h_Q = n$ and is achieved by the same reference strategy. Specifically, when the external parties perform an even number of measurements (even $m$), the simultaneous maximal violation of both inequalities completely characterizes the reference strategy up to the standard self-testing equivalences, namely local isometries, the addition of ancillary degrees of freedom, and complex conjugation of the states and measurements. For odd values of $m$, however, the intrinsic structure of quantum mechanics gives rise to an additional, unavoidable, equivalence corresponding to partial transposition. A detailed discussion of this issue, as well as the formal self-testing theorem, are provided in Appendix~\ref{appendix:self_testing_statement}.

Importantly, the feasibility of self-testing with ROCN inequalities is governed by a purely algebraic property of the coefficient matrix $h$, encoded in an auxiliary matrix $M$ built from pairwise products $h_{xy}h_{zy}$. The full column rank of $M$ is exactly the condition that turns the maximal violation into a complete certification statement.

\begin{theorem}\label{th:self_testing_statement}
    Let $h$ be an $m \times n$ ROCN matrix. Define the $n \times m(m-1)/2$ matrix $M$ as
    \begin{equation}\label{eq:matrix_M}
        M_{y,(x,z)} = h_{xy} h_{zy},
    \end{equation}
    where $y \in \{1,\dots,n\}$ and $x,z \in \{1,\dots,m\}$ with $x < z$, and each pair $(x,z)$ is treated as a single composite index. The maximal violation of both the corresponding ROCN bilocal inequality~\eqref{eq:S_h_inequality} and the linear ROCN Bell inequality~\eqref{eq:rocn_bell_inequality} self-tests the reference strategy from Definition~\ref{def:reference_strategy} if and only if $M$ has full column rank:
    \begin{equation}\label{eq:self-testing_condition}
        \rank(M) = \frac{1}{2}m(m-1).
    \end{equation}
\end{theorem}
The detailed proof of this theorem is provided in Appendix~\ref{appendix:self_testing_statement}. It adapts the framework developed for the bipartite scenario in Ref.~\cite{Michalski2025}; see also~\cite{Konderak2026} for a general construction of self-testing ROCN matrices.

\textit{Examples.}---We now illustrate the construction with several examples. First, observe that the 22-type BRGP inequality~\eqref{eq:BRGP} can be readily recovered within the discussed framework. Indeed, the BRGP case corresponds to choosing the $2 \times 2$ ROCN matrix 
\begin{equation}
    h_\mathrm{BRGP} = \frac{1}{\sqrt{2}} \begin{bmatrix}
        +1 & -1 \\
        +1 & +1
    \end{bmatrix},
\end{equation}
which leads to the quantum bound $\beta^\mathrm{BRGP}_\mathrm{Q} = 2$. The corresponding classical (bilocal) value for this inequality is $\beta^\mathrm{BRGP}_\mathrm{BL} = \sqrt{2}$.

A second paradigmatic example concerns the following ROCN matrix which corresponds to the famous Gisin's elegant Bell inequality~\cite{Gisin2009}:
\begin{equation}\label{EBI}
        h_{\mathrm{EBI}}=\frac{1}{\sqrt 3}
    \begin{bmatrix}
         +1 & +1 & -1 & -1 \\
         +1 & -1 & +1 & -1 \\
         +1 & -1 &  -1 &  +1 
    \end{bmatrix}.
\end{equation}
The associated bilocal bound reads $\mathcal{S}_\mathrm{EBI} \leqslant 2 \sqrt{3}$, while the quantum bound reaches $\beta^\mathrm{EBI}_\mathrm{Q} = 4$. Crucially, this inequality enables self-testing, as its associated matrix $M$ satisfies the full-rank condition~\eqref{eq:self-testing_condition}.

Interestingly, we can obtain a broader class of nonlinear inequalities by considering the family of ROCN matrices generated via the transformation $h_\mathrm{EBI} \mapsto R\, h_\mathrm{EBI}$, where $R$ is a real $3\times3$ matrix that preserves the ROCN property from Definition~\ref{def:ROCN_matrix}. As discussed in Appendix~\ref{appendix:new_family}, the lowest bilocal bound for this family is given by $\beta_\mathrm{BL}^\mathrm{min} = 2\sqrt{2}$. In particular, consider the following one-parameter subclass:
\begin{equation}
    h_\theta = \frac{1}{\sqrt{2}} \begin{bmatrix}
        -1 & +1 & -1 & +1\\
        +\sin \theta & -\cos \theta & -\sin \theta & +\cos \theta\\
        -\cos \theta & -\sin \theta & -\cos\theta & +\sin \theta
    \end{bmatrix},
\end{equation}
where $\theta \in \left[0,\pi/2\right[$. The Bell inequalities derived from these matrices enable self-testing throughout the entire specified parameter regime except at $\theta = 0$, where condition~\eqref{eq:self-testing_condition} fails to hold. Nevertheless, because self-testing is guaranteed for all $\theta > 0$, we can construct self-testing matrices that arbitrarily approach the absolute minimum bilocal bound. In addition, as illustrated in Fig.~\ref{fig:bound_plot}, their bilocal bound $\beta^\theta_\mathrm{BL}$ is always strictly smaller than $\beta^\mathrm{EBI}_\mathrm{BL} = 2\sqrt{3}$, while the quantum bound remains invariant at $4$. Thus, this entire class of inequalities provides a larger quantum-classical gap than the standard elegant Bell inequality~\eqref{EBI}.

\begin{figure}
    \centering
    \includegraphics[width=\linewidth]{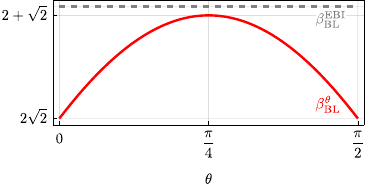}
    \caption{Bilocal bound $\beta^\theta_\mathrm{BL}$ for the family of bilocal ROCN Bell inequalities generated by $h_\theta$ as a function of the parameter $\theta$. The dashed gray line indicates the value of the bilocal bound $\beta_\mathrm{BL}^\mathrm{EBI} = 2\sqrt{3}$ corresponding to the matrix $h_\mathrm{EBI}$. }
    \label{fig:bound_plot}
\end{figure}


\textit{Conclusion.}---In this work, we develop a general and versatile framework for constructing nonlinear Bell inequalities tailored to quantum networks with three parties and two independent sources of correlations. This approach accommodates an arbitrary number of binary measurements per party and, as a special case, recovers the pioneering BRGP inequality introduced in Ref.~\cite{Branciard2012}. The construction is based on ROCN matrices, recently considered in Ref.~\cite{Michalski2025} for deriving Bell inequalities in the standard Bell scenario. A key advantage of this class of matrices is that it enables the analytical determination of the maximal quantum values for the resulting Bell inequalities, which is essential for device-independent applications.

Furthermore, we demonstrate that these inequalities allow for self-testing of the underlying quantum networks. Specifically, we prove that under a specific algebraic condition on the underlying ROCN matrices, the simultaneous maximal violation of both the nonlinear and associated linear inequalities certifies the reference quantum strategy. This strategy consists of a pair of maximally entangled states of suitable local dimensions, complemented by sets of Clifford observables for the external parties. The proof employs an SOS decomposition adapted to the bilocal structure to derive the anticommutation constraints that determine the form of the measurements and enable the reconstruction of the Schmidt coefficients of the shared states.

A number of open directions remain to be explored. Most notably, although explicit constructions of ROCN matrices satisfying the rank condition are known, existing approaches are not optimal in terms of the required number of columns. Improving this scaling behavior would significantly enhance the practical feasibility of our inequalities, potentially revealing tighter connections between the underlying algebraic structure and network nonlocality. Several additional avenues for further development can be identified: (i) designing self-testing schemes for joint entangled measurements performed by the central party beyond the Bell state measurement \cite{Kaniewski2018}, which are more experimentally accessible than the general framework proposed in Ref. \cite{Sarkar2026}; (ii) generalizing this approach to scenarios involving observables with more than two outcomes, which are largely unexplored in the literature; and (iii) developing inequalities tailored to partially entangled states emitted by the sources.

\textit{Acknowledgements.}---This work is supported by the National Science Centre (Poland) through
the SONATA BIS project No. 019/34/E/ST2/00369. The Center for Quantum-Enabled Computing project is carried out within the International Research Agendas programme of the Foundation for Polish Science co-financed by the European Union under the European Funds for Smart Economy 2021-2027 (FENG).
\bibliography{bibliography}

\onecolumngrid
\appendix

\section{Quantum bound}\label{sec:quantum_bound}
In this appendix, we provide the proof of the quantum bound for ROCN bilocal inequalities and their linear counterparts required for self-testing.
\begin{proof}[Proof of Theorem~\ref{th:quantum_bound}]
    For the given ROCN matrix $h$, introduce the shorthand notation
    \begin{equation}\label{eq:A_C_slashed}
        \slashed{A}_y = \sum_{x=1}^m h_{xy} A_x, \quad \slashed{C}_y = \sum_{z=1}^m h_{zy} C_z.
    \end{equation}
    Then, the family of operators $\mathcal{B}_y = \slashed{A}_y \otimes B_y \otimes \slashed{C}_y$ reproduces the quantities $I_y$ appearing in Eq.~\eqref{eq:I_y}:
    \begin{equation}
        I_y = \langle \mathcal{B}_y \rangle = \sum_{x=1}^m \sum_{z=1}^m h_{xy} h_{zy} \langle A_x B_y C_z \rangle.
    \end{equation}
    Since the square root function is a concave function, Jensen's inequality gives
    \begin{equation}\label{eq:Jensen}
        \frac{1}{n} \mathcal{S}_h = \frac{1}{n} \sum_{y=1}^n \sqrt{|I_y|} \leqslant \sqrt{\frac{1}{n} \sum_{y=1}^n \left|I_y\right|},
    \end{equation}
    with equality if and only if $|I_1| = |I_2| = \ldots = |I_n|$. We want to upper bound the sum appearing under the square root on the right hand side of the above inequality. This can be done via a suitable SOS decomposition. Take $\vec{s}\in\{-1,1\}^n$ and consider the operators
    \begin{equation}
       N_{y,\vec{s}} \vcentcolon= \slashed{A}_y \otimes \mathds{1} \otimes \mathds{1} -  s_y \mathds{1} \otimes B_y \otimes \slashed{C}_y.
    \end{equation}
    A direct computation, using the ROCN conditions for the matrix $h$, yields
    \begin{align}
        \frac{1}{2} \sum_{y=1}^n N^\dag_{y,\vec{s}} N_{y,\vec{s}} &= \frac{1}{2} \sum_{y=1}^n \left(\slashed{A}_y \otimes \mathds{1} \otimes \mathds{1} - s_y \mathds{1} \otimes B_y \otimes \slashed{C}_y\right)^2 \nonumber\\
        &= \frac{1}{2}\sum_{y=1}^n \left[\sum_{x=1}^m  \sum_{z=1}^m h_{xy}h_{zy} \left(A_x A_z \otimes \mathds 1\otimes \mathds 1+\mathds 1\otimes \mathds 1\otimes C_xC_z\right)-2 s_y \mathcal B_y\right] \nonumber\\ &=\sum_{y=1}^n\left[\sum_{x=1}^m h_{xy}^2\, \mathds{1} \otimes \mathds{1} \otimes \mathds{1}-s_y\mathcal{B}_y\right] =n \mathds{1} \otimes \mathds{1} \otimes \mathds{1} - \sum_{y=1}^n s_y\mathcal{B}_y \geqslant 0,
    \end{align}
    where we further employed the conditions $A_x^2=B_y^2=C_z^2=\mathds 1$.
    Taking the expectation value in an arbitrary state $\ket{\psi}$ gives
    \begin{equation}\label{eq:linear_bound}
        \sum_{y=1}^n s_y I_y = \sum_{y=1}^n s_y \langle \mathcal{B}_y \rangle \leqslant n.
    \end{equation}
    Observe that the sign of each term in the sum can be changed independently by redefining the vector $\vec{s}$. Since the bound above is invariant under such redefinitions, it follows that
    \begin{equation}\label{eq:B_quantum_bound}
        \sum_{y=1}^n |I_y| \leqslant n,
    \end{equation}
    which corresponds to taking $s_y = \mathrm{sgn}(I_y)$. Substituting Eq.~\eqref{eq:B_quantum_bound} into Eq.~\eqref{eq:Jensen} results in the following upper bound on the maximal quantum value:
    \begin{equation}\label{eq:S_quantum_bound}
        \beta_Q^h = \max_{\ket{\psi},\{A_x\},\{B_y\},\{C_z\}} \sum_{y=1}^n \sqrt{\left|\bra{\psi} {\mathcal{B}}_y \ket{\psi}\right|} \leqslant n.
    \end{equation}
    Inequality \eqref{eq:S_quantum_bound} is saturated exactly when both \eqref{eq:Jensen} and \eqref{eq:B_quantum_bound} are saturated. The latter occurs if and only if there exists a state $\ket{\psi}$ lying in the kernel of all $N_{y,\vec{s}}$ with $s_y = \mathrm{sgn}(I_y)$, which is equivalent to the condition
    \begin{equation}\label{eq:saturability}
        \forall y\in\{1,\ldots,n\}:\quad \left(\slashed{A}_y \otimes \mathds{1} \otimes \mathds{1}\right) \ket{\psi} =  \mathrm{sgn}(I_y) \left(\mathds{1} \otimes B_y \otimes \slashed{C}_y\right) \ket{\psi}.
    \end{equation}
    Saturation of \eqref{eq:Jensen} additionally requires that all the absolute values of combinations $I_y$ coincide, $|I_1| = |I_2| = \ldots = |I_n|$.
    
    To verify that the bound~\eqref{eq:S_quantum_bound} is tight, consider the reference strategy introduced in Definition~\ref{def:reference_strategy}. Using the anticommutation relations together with the ROCN conditions, we obtain
    \begin{align}
        {\slashed{{A}}}_y^2 = \sum_{x=1}^m \sum_{z=1}^m h_{xy} h_{zy} \tilde A_x \tilde A_z = \sum_{x=1}^m h_{xy}^2 \tilde A_x^2 + \sum_{x < z}^m h_{xy} h_{zy} \{\tilde A_x,\tilde A_z\} = \mathds{1}.
    \end{align}
    The same result holds for Charlie. The saturation condition~\eqref{eq:saturability} can be rewritten as
    \begin{equation}
        \forall y\in\{1,\ldots,n\}:\quad \mathrm{sgn}(I_y) \left(  \slashed{A}_y \otimes B_y \otimes \slashed{C}_y \right) \ket{\psi} = \ket{\psi}.
    \end{equation}
    Next, we employ the standard flip identity
    \begin{equation}\label{eq:flip-flop}
        \mathcal{O}\otimes \mathds{1} \,\ket{\Phi_d} = \mathds{1} \otimes \mathcal{O}^\mathsf{T}\ket{\Phi_d},
    \end{equation}
    valid for any operator $\mathcal{O} \in\mathcal B(\tilde{\mathcal H}_\mathrm{A})$. Applying this twice, we obtain
    \begin{align}\label{eq:kernel_N}
        \left(\slashed{A}_y \otimes B_y \otimes \slashed{C}_y \right) \ket{\Phi_d} \otimes \ket{\Phi_d} &= \left(\slashed{A}_y \otimes (\slashed{A}^\mathsf{T}_y \otimes \slashed{C}^\mathsf{T}_y) \otimes \slashed{C}_y \right) \ket{\Phi_d} \otimes \ket{\Phi_d}  \nonumber\\
        &= \left(\slashed{A}_y^2 \otimes \mathds{1} \otimes \mathds{1} \otimes \slashed{C}_y^2 \right) \ket{\Phi_d} \otimes \ket{\Phi_d} = \ket{\Phi_d} \otimes \ket{\Phi_d}.
    \end{align}
    Thus, $\ket{\Phi_d} \otimes \ket{\Phi_d}$ lies in the kernel of every $N_{y,\vec{s}}$ with $s_y = \mathrm{sgn}(I_y) = 1$, showing that the SOS bound~\eqref{eq:B_quantum_bound} is saturated. Moreover, the reference strategy satisfies $I_1 = I_2 = \ldots = I_n = 1$, so the inequality~\eqref{eq:Jensen} is also saturated. We therefore conclude that $\beta_Q^h=n$.
\end{proof}
Observe that the above arguments also establish that the quantum bound of the linear expression~\eqref{eq:rocn_bell_inequality} is $\beta_Q^h = n$. Indeed, the combination $\mathcal{L}_h$ is a particular case of the general expression in Eq.~\eqref{eq:linear_bound} obtained by setting $s_1 = s_2 = \ldots = s_n = 1$, and is thus subject to the same upper bound. Moreover, the reference strategy attains $\mathcal{L}_h = n$, so this bound is saturated.

\section{Jordan--Wigner representation}\label{sec:jordan_wigner}
In this appendix, we revisit the Jordan--Wigner representation theorem, which provides an explicit construction of a collection of mutually anticommuting observables.

\begin{theorem}
    Let $\{\Gamma_x\}_{x=1}^m$ be a set of observables acting on a finite-dimensional Hilbert space $\mathcal{H}$ and satisfying the canonical anticommutation relations $\{\Gamma_x, \Gamma_z \} = 2\delta_{xz} \mathds{1}$. Then, up to a unitary transformation, the Hilbert space $\mathcal H$ can be decomposed as a tensor product of qubits:
    \begin{equation}\label{eq:tensor_product_JW}
        \mathcal H \cong \bigotimes_{k=1}^{r} \mathbb{C}^2 \otimes \mathcal H',
    \end{equation}
    where $r=\lfloor m/2 \rfloor$. With respect to this decomposition, the first $2r$ operators $\Gamma_x$ can be expressed in terms of Pauli matrices as    
    \begin{equation}\label{eq:U}
        \Gamma_x =
        \begin{cases}
            Y^{\otimes (x-1)/2} \otimes Z \otimes \mathds{1} & \text{for } x \text{ odd}, \\
            Y^{\otimes (x-2)/2} \otimes X \otimes \mathds{1} & \text{for } x \text{ even},
        \end{cases}
    \end{equation}
    with the identities acting on the remaining qubit subspaces and $\mathcal{H}'$. For $m$ odd, the last operator is in the form
    \begin{equation}\label{eq:last_operator_JW}
        \Gamma_m = Y^{\otimes r} \otimes \Gamma_m'.
    \end{equation}
    where $\Gamma_m'$ is a Hermitian and unitary operator acting on $\mathcal{H'}$.
\end{theorem}

\section{Bilocal bound}\label{appendix:bilocal_bound}
In this appendix, we recall the conditions on the ROCN matrix which guarantee that the resulting ROCN bilocal inequality exhibits a nontrivial gap between the bilocal (classical) and quantum bounds. 

\begin{theorem}\label{th:bilocal_bound}
	For any ROCN matrix $h\in\mathds{R}^{m \times n}$, the bilocal bound of the corresponding ROCN bilocal inequality~\eqref{eq:S_h_inequality} is
	\begin{equation}\label{eq:classical_bound}
		\beta^h_\mathrm{BL} = \max_{\vec{a} \in \{-1,1\}^m} \sum_{y = 1}^n \left| \sum_{x = 1}^m h_{xy} a_x \right|.
	\end{equation}
	In order to ensure a quantum violation of the bilocal Bell inequality, $\beta^h_\mathrm{BL} < \beta^h_\mathrm{Q}$, it is necessary and sufficient to impose the following condition on the matrix $h$:
	\begin{equation}\label{eq:non-triviality}
		\forall \vec{a} \in \{-1,1\}^m \; \exists y \in \{1,\ldots,n\}: \quad \sum_{x < z}^m h_{xy} h_{zy} a_x a_z \ne 0.
	\end{equation}
\end{theorem}

\begin{proof}
    For bilocal correlations, each combination $I_y$ admits the decomposition
    \begin{align}
        I_y = \int \mathrm{d}\lambda_1\, \mu_1(\lambda_1) \int \mathrm{d}\lambda_2\, \mu_2(\lambda_2) \sum_{x=1}^m\sum_{z=1}^m h_{xy} h_{zy} \langle A_x \rangle_{\lambda_1} \langle B_y \rangle_{\lambda_1,\lambda_2} \langle C_z \rangle_{\lambda_2}.
    \end{align}
    Using $|\langle B_y \rangle_{\lambda_1,\lambda_2}| \leqslant 1$ and applying the Cauchy-Schwarz inequality gives:
    \begin{align}
        |I_y| &\leqslant \int \mathrm{d}\lambda_1\, \mu_1(\lambda_1) \int \mathrm{d}\lambda_2\, \mu_2(\lambda_2) \abs{\sum_{x=1}^m h_{xy} \langle A_x \rangle_{\lambda_1}} \abs{\sum_{z=1}^m h_{zy} \langle C_z \rangle_{\lambda_2}} \nonumber\\
        &\leqslant \int \mathrm{d}\lambda_1\, \mu_1(\lambda_1) \abs{\sum_{x=1}^m h_{xy} \langle A_x \rangle_{\lambda_1}} \times \int \mathrm{d}\lambda_2\, \mu_2(\lambda_2) \abs{\sum_{z=1}^m h_{zy} \langle C_z \rangle_{\lambda_2}}.
    \end{align}
    As in the standard bipartite Bell scenario, any intrinsic randomness in $\langle A_x \rangle_{\lambda_1}$ and $\langle C_z \rangle_{\lambda_2}$ can be absorbed into $\mu_1(\lambda_1)$ and $\mu_2(\lambda_2)$, respectively. Therefore, without loss of generality, we may assume deterministic strategies $\langle A_x \rangle_{\lambda_1},\langle C_z \rangle_{\lambda_2}\in\{-1,+1\}$ for any $\lambda_1$, $\lambda_2$. 
    
    We now bound $\mathcal{S}_h$ using the Cauchy-Schwarz inequality:
    \begin{align}
        \mathcal{S}_h = \sum_{y=1}^n \sqrt{|I_y|} \leqslant \sqrt{\int \mathrm{d}\lambda_1\, \mu_1(\lambda_1) \sum_{y=1}^n \abs{\sum_{x=1}^m h_{xy} \langle A_x \rangle_{\lambda_1}}} \times \sqrt{\int \mathrm{d}\lambda_2\, \mu_2(\lambda_2) \sum_{y=1}^n \abs{\sum_{z=1}^m h_{zy} \langle C_z \rangle_{\lambda_2}}}.
    \end{align}
    Each integrand is bounded by the maximum value over deterministic assignments $\vec{a} \in \{-1,1\}^m$. Consequently, the bilocal bound reads:
    \begin{equation}
        \beta_\mathrm{BL} = \max_{\vec{a} \in \{-1,1\}^m} \sum_{y = 1}^n \left| \sum_{x = 1}^m h_{xy} a_x \right|.
    \end{equation}
    The condition for a quantum violation then follows directly from the analysis of bipartite ROCN Bell inequalities, as established in Ref.~\cite{Michalski2025}.
\end{proof}

It is worth noting that condition~\eqref{eq:non-triviality} can be explicitly phrased in terms of the matrix $M$ defined in Eq.~\eqref{eq:matrix_M}. Specifically, while Eq.~\eqref{eq:self-testing_condition} guarantees that only anticommuting observables can saturate the quantum bound, Eq.~\eqref{eq:classical_bound} ensures that this bound remains strictly unattainable by any bilocal deterministic strategy.

\section{Self-testing of the reference strategy}\label{appendix:self_testing_statement}

In this appendix, we recall the formal definition of self‑testing introduced in Ref.~\cite{Supic2020}. Self-testing represents the strongest form of device‑independent certification. It guarantees that a given correlation not only witnesses nonlocality, but also uniquely identifies---up to the unavoidable physical equivalences---a reference quantum strategy.

\begin{definition}[Bilocal self-testing]\label{def:bilocal_self-testing}
    Let $\vec{c} = \{p(a,b,c \vert x,y,z)\}$ be a quantum correlation generated by a state $\ket*{\tilde\psi}=\ket*{\tilde \psi}_{\mathrm{AB_1}}\otimes \ket*{\tilde \psi}_{\mathrm{B_2C}}\in\tilde{\mathcal H}_{\mathrm A} \otimes \tilde{\mathcal H}_{\mathrm{B}_1}\otimes \tilde{\mathcal H}_{\mathrm B_2}\otimes \tilde{\mathcal H}_{\mathrm C}$, and measurements  $\{\tilde A_x\}$, $\{\tilde B_y\}$ and $\{\tilde C_z\}$, with $x\in\{1,\ldots,n_\mathrm A\},\ y\in\{1,\ldots,n_\mathrm B\}$ and $z\in\{1,\ldots,n_\mathrm C\}$. We say that $\vec{c}$ self-tests this strategy if, for every state $\ket{\psi}_{\mathrm{AB_1}}\otimes\ket{\psi}_{\mathrm{B_2C}}\in\mathcal{H}_\mathrm{A}\otimes\mathcal{H}_{\mathrm{B}_1}\otimes\mathcal{H}_{\mathrm{B}_2}\otimes \mathcal{H}_\mathrm{C}$ and every observables $\{A_x\}$, $\{B_y\}$ and $\{C_z\}$ generating $\vec{c}$, there exist local isometries
        \begin{equation}\label{eq:self-testing_isometry_standard}   U_{\mathrm{A}}:\mathcal{H}_\mathrm{A}\to\tilde{\mathcal H}_{\mathrm A}\otimes\mathcal{H}_{\mathrm{A}'},\quad       U_{\mathrm{B}}:\mathcal{H}_\mathrm{B}\to(\tilde{\mathcal H}_{\mathrm B_1}\otimes\tilde{\mathcal H}_{\mathrm B_2})\otimes\mathcal{H}_{\mathrm{B}'},\quad        U_{\mathrm{C}}:\mathcal{H}_\mathrm{C}\to\tilde{\mathcal H}_{\mathrm C}\otimes\mathcal{H}_{\mathrm{C}'},
    \end{equation}
    and an auxiliary state $\ket{\xi}\in\mathcal{H}_{\mathrm{A}'}\otimes\mathcal{H}_{\mathrm{B}'}\otimes\mathcal{H}_{\mathrm{C}'}$ such that, for all $x,y,z$,
    \begin{equation}\label{eq:bilocal_self_testing_def_standard}
        \left(U_{\mathrm{A}}\otimes U_{\mathrm{B}}\otimes U_{\mathrm{C}}\right)
        \left(A_x\otimes B_y\otimes C_z\right)
        \left(\ket{\psi}_{\mathrm{AB_1}}\otimes\ket{\psi}_{\mathrm{B_2C}}\right)
        =
        \left(\tilde  A_x\otimes\tilde B_y\otimes\tilde C_z\right)\ket*{\tilde\psi}\otimes\ket{\xi},
    \end{equation}
    with $x\in\{0,\dots, n_\mathrm{A}\}$, $y\in\{0,\dots, n_\mathrm{B}\}$ and $z\in\{0,\dots, n_\mathrm{C}\}$. In particular, we employed the notation $A_0=B_0=C_0=\mathds 1$.
\end{definition}

Consider the reference strategy from Definition~\ref{def:reference_strategy}, where $n_{\mathrm A}=n_\mathrm{C}=m$ and $n_\mathrm B=n$. Here, Alice, Bob, and Charlie share two independent copies of the maximally entangled state~\eqref{eq:me_state} distributed across the bipartitions $(\tilde{\mathcal{H}}_{\mathrm A}, \tilde{\mathcal{H}}_{\mathrm B_1})$ and $(\tilde{\mathcal{H}}_{\mathrm B_2},\tilde{\mathcal{H}}_{\mathrm C})$. Alice's and Charlie's observables are defined as
\begin{equation}
	\forall x\in\{1,\dots, m\}:\quad \tilde A_x=\tilde C_x=\Gamma_x,
\end{equation} 
where the operators $\Gamma_x$ act on the Hilbert space $\mathbb C^d\equiv\mathbb C^{2\otimes r}$ and are given by Eqs.~\eqref{eq:U} and~\eqref{eq:last_operator_JW} of Appendix~\ref{sec:jordan_wigner} (for $\mathcal H'=\mathbb C$ and $\Gamma_m'=1$). Bob's measurements are defined according to Eq.~\eqref{eq:bob_observables}. 

Now, suppose $m$ is odd. If we flip the sign of the final operator, $\tilde A_m\rightarrow -\tilde A_m$, Alice's modified set of observables $\{\tilde A_1,\dots,\tilde A_{m-1},-\tilde A_m\}$ still satisfies the anticommutation relations. This sign-flip transformation can be implemented by applying a partial transposition $\Phi$ to the last qubit of Alice's subspace. Crucially, if the same transformation is simultaneously applied to Bob's first Hilbert space $\tilde{\mathcal{H}}_{\mathrm B_1}$, the overall observed correlations $\{p(a,b,c|x,y,z)\}$ remain unchanged. 

It was shown in Ref.~\cite{Michalski2025} that the original reference strategy and its partially transposed counterpart cannot be related to one another via a unitary transformation---nor via complex conjugation when $r=\lfloor \frac{m}{2}\rfloor$ is even. Consequently, when $m$ is odd, the reference strategy cannot be self-tested in the canonical sense of Definition~\ref{def:bilocal_self-testing}, but only up to a partial transposition. Specifically, Eq.~\eqref{eq:bilocal_self_testing_def_standard} must be replaced with
\begin{equation}\label{eq:bilocal_self_testing_def}
    \left(U_{\mathrm{A}}\otimes U_{\mathrm{B}}\otimes U_{\mathrm{C}}\right)
    \left(A_x\otimes B_y\otimes C_z\right)
    \left(\ket{\psi}_{\mathrm{AB_1}}\otimes\ket{\psi}_{\mathrm{B_2C}}\right)
    =
    \left(\bar A_x\otimes\bar B_y\otimes\bar C_z\right)\ket*{\tilde\psi}\otimes\ket{\xi}.
\end{equation}
The observables $\bar A_x$ and $\bar C_z$ take the form
\begin{equation}
    \bar A_x=\tilde A_x\otimes M_0+\Phi(\tilde A_x)\otimes M_1,
    \quad
    \bar C_z=\tilde C_z\otimes N_0+\Phi(\tilde C_z)\otimes N_1,
\end{equation}
while $\bar B_y$ decomposes as
\begin{equation}
    \bar B_y =\tilde B_y \otimes P_0  Q_0+(\Phi\otimes \mathds 1)(\tilde B_y) \otimes P_1 Q_0+ (\mathds 1\otimes\Phi)(\tilde B_y) \otimes P_0 Q_1+(\Phi\otimes \Phi)(\tilde B_y) \otimes P_1Q_1.
\end{equation}
where $\Phi$ denotes the additional ROCN equivalence map introduced in Ref.~\cite{Michalski2025}. The operators $\{M_0,M_1\}$, $\{N_0,N_1\}$, $\{P_0,P_1\}$ and $\{Q_0,Q_1\}$ form binary projection-valued measures (PVMs) satisfying
\begin{equation}
    \mel{\xi}{M_0\otimes P_0\otimes \mathds 1+M_1\otimes P_1\otimes \mathds 1}{\xi}=1,\quad \mel{\xi}{\mathds 1\otimes Q_0\otimes N_0 +\mathds 1\otimes Q_1\otimes N_1}{\xi}=1.
\end{equation}

\section{Proof of self-testing}\label{sec:proof_self_testing}
\begin{proof}[Proof of Theorem~\ref{th:self_testing_statement}]
We prove that condition~\eqref{eq:self-testing_condition} implies that the maximal violation of the ROCN bilocal inequality~\eqref{eq:S_h_inequality} and its linear counterpart~\eqref{eq:rocn_bell_inequality} self-tests the reference strategy. The converse direction follows from the considerations presented in Ref.~\cite[Lemma 1]{Michalski2025}. The proof is organized into three parts. In \emph{Part~I}, we show that achieving the maximal quantum violation of the nonlinear inequality forces Alice's and Charlie's observables to form pairwise anticommuting families. In \emph{Part~II}, we employ vectorization techniques to establish relations among Alice's, Charlie's, and Bob's observables, as well as among the Schmidt coefficients of the shared quantum states. The linear inequality is subsequently used to determine the sign of Bob's observables. Finally, in \emph{Part~III}, we invoke the Jordan--Wigner representation theorem to demonstrate self-testing of the reference strategy.
    
    \emph{Part~I.} Suppose that a strategy $\{\ket{\psi}, A_x, B_y, C_z\}$ achieves the maximal quantum violation of the considered ROCN inequalities~\eqref{eq:S_h_inequality} and~\eqref{eq:rocn_bell_inequality}. We assume without loss of generality~\cite{Baptista2023} that all observables $A_x$, $B_y$, and $C_z$ square to the identity and that the shared state $\ket{\psi}$ is pure. For the given ROCN matrix $h$, we introduce again the shorthand notation~\eqref{eq:A_C_slashed}.  Consider the operators
    \begin{equation}
       N_{y}= \slashed{A}_y \otimes \mathds{1} \otimes \mathds{1} - \mathrm{sgn}(I_y)\, \mathds{1} \otimes B_y \otimes \slashed{C}_y,
    \end{equation}
    with $I_y$ defined as in Eq.~\eqref{eq:I_y}, taking the expectation value of $\mathcal{B}_y$ in the state $\ket{\psi}$. A direct computation using the ROCN conditions for the matrix $h$ yields
    \begin{align}\label{eq:self-testing_SOS}
        \frac{1}{2} \sum_{y=1}^n N_y^\dagger N_y &= n \mathds{1} \otimes \mathds{1} \otimes \mathds{1} - \sum_{y=1}^n \mathrm{sgn}(I_y)\, \mathcal{B}_y \geqslant 0.
    \end{align}
    Taking the expectation value in the state $\ket{\psi}$ gives
    \begin{equation}
        \sum_{y=1}^n |I_y| = \sum_{y=1}^n \mathrm{sgn}(I_y)\, \langle \mathcal{B}_y \rangle \leqslant n.
    \end{equation}
    From the SOS decomposition~\eqref{eq:self-testing_SOS} and the saturation condition for inequality~\eqref{eq:Jensen}, it follows that the maximal quantum value $\beta_Q^{h} = n$ is attained if and only if the shared state $\ket{\psi}$ satisfies $N_y\ket{\psi}=0$ for all $y\in\{1,\dots, n\}$ and all absolute values $|I_y|$ coincide, i.e., $|I_1| = |I_2| = \ldots = |I_n| = 1$. As a consequence, for all $y\in\{1,\dots, n\}$ we have
    \begin{align}\label{eq:saturation_network_bound}
        \left(\slashed{A}_y^2 \otimes \mathds{1} \otimes \mathds{1} \right) \ket{\psi} = \mathrm{sgn}(I_y)\, \left( \slashed{A}_y \otimes B_y \otimes \slashed{C}_y \right) \ket{\psi} = \left(  \mathds{1} \otimes \mathds{1} \otimes \slashed{C}_y^2 \right) \ket{\psi}.
    \end{align}
    Acting from the left with $\bra{\psi}$ gives
    \begin{align}
        \bra{\psi} \left(\slashed{A}_y^2 \otimes \mathds{1} \otimes \mathds{1} \right) \ket{\psi} = \bra{\psi} \left( \mathds{1} \otimes \mathds{1} \otimes \slashed{C}_y^2 \right) \ket{\psi} = 1.
    \end{align}
    Expanding the operators $\slashed{A}_y$ and $\slashed{C}_y$, using $A_x^2 = C_z^2 = \mathds{1}$, and applying the ROCN conditions for the matrix $h$ yields
    \begin{equation}
        \sum_{x<z}^m h_{xy} h_{zy} \ev{\{A_x, A_z\}}{\psi} = \sum_{x<z}^m h_{xy} h_{zy} \ev{\{C_x, C_z\}}{\psi} = 0.
    \end{equation}
    Generalizing the discussion in Ref.~\cite[Appendix C]{Michalski2025}, we can additionally assume that the reduced density matrices $\rho_{\mathrm{A}}$, $\rho_{\mathrm{B}_1}$, $\rho_{\mathrm{B}_2}$, and $\rho_\mathrm{C}$ are full rank.
    Under this assumption, the above relation implies
    \begin{align}\label{eq:eqnset}
        \sum_{x<z}^m h_{xy} h_{zy} \{A_x, A_z\} = \sum_{x<z}^m h_{xy} h_{zy} \{C_x, C_z\} = 0.
    \end{align}
    If the matrix $M$ defined in Eq.~\eqref{eq:matrix_M} has full column rank, then Eq.~\eqref{eq:eqnset} enforces that Alice's and Charlie's observables satisfy the canonical anticommutation relations.

    \emph{Part~II.} Due to source independence, the total shared quantum state $\ket{\psi}$ can be decomposed as a product $\ket{\psi} = \ket{\psi}_{\mathrm{AB_1}} \otimes \ket{\psi}_{\mathrm{B_2C}}$. Let us express the state $\ket{\psi}_{\mathrm{AB_1}}$ shared between Alice and Bob in its Schmidt decomposition:
    \begin{equation}
        \ket{\psi}_{\mathrm{AB_1}} = \sum_{k=1}^{D_\mathrm{A}} \lambda_k \ket{k} \otimes \ket{k} =  (P_\mathrm{A} \otimes \mathds{1}) \ket{\Phi_{D_\mathrm{A}}},
    \end{equation}
    where $P_\mathrm{A} = \sqrt{D_\mathrm{A}} \sum_{k=1}^{D_\mathrm{A}} \lambda_k \ketbra{k}$, $D_\mathrm{A}=\dim \mathcal{H}_{\mathrm{A}}=\dim\mathcal{H}_{\mathrm{B}_1}$ and
    \begin{equation}\label{eq:me_state_AB1}
        \ket{\Phi_{D_\mathrm{A}}} =\frac{1}{\sqrt{D_\mathrm{A}}} \sum_{k=1}^{D_\mathrm{A}} \ket{k}\otimes\ket{k}
    \end{equation}
    is the maximally entangled state between $\mathcal{H}_\mathrm{A}$ and $\mathcal{H}_{\mathrm{B}_1}$. Observe that, since we assumed the reduced density matrices $\rho_{\mathrm{A}}$ and $\rho_{\mathrm{B}_1}$ to be full rank, the Hilbert spaces $\mathcal{H}_{\mathrm{A}}$ and $\mathcal{H}_{\mathrm{B}_1}$ have the same dimensions. Moreover, we choose the basis so that it coincides with the eigenbasis of both $\rho_{\mathrm{A}}$ and $\rho_{\mathrm{B}_1}$.

    Similarly, we decompose the state $\ket{\psi}_{\mathrm{B_2C}}$ as
    \begin{equation}
        \ket{\psi}_{\mathrm{B_2C}} = (\mathds{1} \otimes P_\mathrm{C}) \ket{\Phi_{D_\mathrm{C}}},
    \end{equation}
    where $P_\mathrm{C} = \sqrt{D_\mathrm{C}} \sum_{k=1}^{D_\mathrm{C}} \mu_k \ketbra{k}$, $D_\mathrm{C}=\dim \mathcal{H}_{\mathrm{C}}=\dim\mathcal{H}_{\mathrm{B}_2}$, and $\ket{\Phi_{D_\mathrm{C}}}$ is the maximally entangled state between $\mathcal{H}_\mathrm{C}$ and $\mathcal{H}_{\mathrm{B}_2}$, defined analogously to Eq.~\eqref{eq:me_state_AB1}.
    
    The anticommutation relations and the ROCN conditions for the matrix $h$ imply that $\slashed{A}_y^2 = \slashed{C}_y^2 = \mathds{1}$ for all $y\in\{1,\dots,n\}$. Furthermore, we can apply the operator $\slashed{A}_y \otimes \mathds{1} \otimes \slashed{C}_y$ to both sides of Eq.~\eqref{eq:saturation_network_bound} and use the flip identity from Eq.~\eqref{eq:flip-flop} to obtain
    \begin{equation}
        \left[\mathds{1} \otimes B_y (P_\mathrm{A} \otimes P_\mathrm{C}) \otimes \mathds{1}\right] \ket{\Phi_{D_\mathrm{A}}} \otimes \ket{\Phi_{D_\mathrm{C}}}\\
        = \mathrm{sgn}(I_y) \left[\mathds{1} \otimes \left(P_\mathrm{A} \slashed{A}^\mathsf{T}_y \otimes P_\mathrm{C} \slashed{C}^\mathsf{T}_y\right) \otimes \mathds{1} \right] \ket{\Phi_{D_\mathrm{A}}} \otimes \ket{\Phi_{D_\mathrm{C}}},
    \end{equation}
    Projecting both sides onto $\ket{\ell_\mathrm{A}} \otimes \left(\ket{k_\mathrm{A}} \otimes \ket{k_\mathrm{C}}\right) \otimes \ket{\ell_\mathrm{C}}$ for all $k_\mathrm{A},\ell_\mathrm{A} \in \{1, \ldots, D_\mathrm{A}\}$ and $k_\mathrm{C},\ell_\mathrm{C} \in \{1, \ldots, D_\mathrm{C}\}$ yields
    \begin{equation}
        \bra{k_\mathrm{A}} \otimes \bra{k_\mathrm{C}} B_y (P_\mathrm{A} \otimes P_\mathrm{C}) \ket{\ell_\mathrm{A}} \otimes \ket{\ell_\mathrm{C}}
        = \mathrm{sgn}(I_y)\bra{k_\mathrm{A}} \otimes \bra{k_\mathrm{C}} \left(P_\mathrm{A} \slashed{A}^\mathsf{T}_y \otimes P_\mathrm{C} \slashed{C}^\mathsf{T}_y\right) \ket{\ell_\mathrm{A}} \otimes \ket{\ell_\mathrm{C}}
    \end{equation}
    and hence, on the subspace available to Bob
    \begin{equation}
        B_y (P_\mathrm{A} \otimes P_\mathrm{C}) = \mathrm{sgn}(I_y)(P_\mathrm{A} \otimes P_\mathrm{C}) \left(\slashed{A}^\mathsf{T}_y \otimes \slashed{C}^\mathsf{T}_y\right).
    \end{equation}
    Taking the Hermitian conjugate gives
    \begin{equation}
        (P_\mathrm{A} \otimes P_\mathrm{C}) B_y = \mathrm{sgn}(I_y)\left(\slashed{A}^\mathsf{T}_y \otimes \slashed{C}^\mathsf{T}_y\right) (P_\mathrm{A} \otimes P_\mathrm{C}).
    \end{equation}
    Combining these two equalities, we find
    \begin{equation}
         (P_\mathrm{A} \otimes P_\mathrm{C})  B_y B_y  (P_\mathrm{A} \otimes P_\mathrm{C}) = (P_\mathrm{A} \otimes P_\mathrm{C})^2
         = \left(\slashed{A}^\mathsf{T}_y \otimes \slashed{C}^\mathsf{T}_y\right) (P_\mathrm{A} \otimes P_\mathrm{C})^2 \left(\slashed{A}^\mathsf{T}_y \otimes \slashed{C}^\mathsf{T}_y\right),
    \end{equation}
    which implies that $[P_\mathrm{A}^2 \otimes P_\mathrm{C}^2, \slashed{A}^\mathsf{T}_y \otimes \slashed{C}^\mathsf{T}_y] = 0$. Since $P_\mathrm{A} \otimes P_\mathrm{C}$ is positive, it follows that $[P_\mathrm{A} \otimes P_\mathrm{C}, \slashed{A}^\mathsf{T}_y \otimes \slashed{C}^\mathsf{T}_y] = 0$, and we obtain
    \begin{equation}\label{eq:operators_bob}
        B_y = \mathrm{sgn}(I_y) \slashed{A}^\mathsf{T}_y \otimes \slashed{C}^\mathsf{T}_y = \mathrm{sgn}(I_y) \sum_{x=1}^m \sum_{z=1}^m h_{xy} h_{zy} A_x^\mathsf{T} \otimes C_z^\mathsf{T}.
    \end{equation}
    This demonstrates that the structure of Bob's observables is fully specified by $A_x$ and $C_z$, apart from an overall sign. This residual freedom is removed by requiring maximal violation of the linear expression~\eqref{eq:rocn_bell_inequality}. Because the combination $\mathcal{L}_h = \sum_y I_y$ attains the maximal quantum value $\beta_Q^h = n$, it follows that $I_1 = I_2 = \ldots = I_n = 1$. 
    
    Finally, the commutation relation $[P_\mathrm{A} \otimes P_\mathrm{C}, \slashed{A}^\mathsf{T}_y \otimes \slashed{C}^\mathsf{T}_y] = 0$ can be rewritten as
    \begin{equation}\label{eq:commutation_rewritten}
        [P_\mathrm{A}, \slashed{A}^\mathsf{T}_y] \otimes  P_\mathrm{C} \slashed{C}^\mathsf{T}_y + \slashed{A}^\mathsf{T}_y P_\mathrm{A} \otimes  [P_\mathrm{C}, \slashed{C}^\mathsf{T}_y] = 0.
    \end{equation}
    Multiplying the operator equation by $\mathds{1} \otimes \slashed{C}_y^\mathsf{T}$, taking the partial trace over $\mathcal{H}_{\mathrm{B}_2}$, and using $\slashed{C}_y^2 = \mathds{1}$ results in
    \begin{equation}
        [P_\mathrm{A}, \slashed{A}^\mathsf{T}_y] \Tr(P_\mathrm{C}) = 0.
    \end{equation}
    Since $P_\mathrm{C}$ is positive definite, by substituting the definition~\eqref{eq:A_C_slashed}, we obtain
    \begin{equation}
        \sum_{x=1}^m h_{xy} [P_\mathrm{A}, A^\mathsf{T}_x] = 0,
    \end{equation}
    for every $y\in\{1,\dots,n\}$. Multiplying this equation by $h_{zy}$ and summing over $y$ leads to
    \begin{equation}
        \forall z\in\{1,\dots,m \}:\quad \sum_{y=1}^n h_{zy}^2\  [P_\mathrm{A}, {A}_z^\mathsf{T}]=0 \implies [P_\mathrm{A}, {A}_z^\mathsf{T}]=0.
    \end{equation}
    Similarly, multiplying Eq.~\eqref{eq:commutation_rewritten} by $\slashed{A}_y^\mathsf{T} \otimes \mathds{1}$, taking the partial trace over $\mathcal{H}_{\mathrm{B}_1}$, and performing analogous steps, we get $[P_\mathrm{C}, C_z^\mathsf{T}] = 0$ for all $z\in\{1,\dots,m\}$.
    
    Altogether, we have established that the sets of Alice's and Charlie's observables each satisfy the canonical anticommutation relations, Bob's observables take the product form $\slashed{A}_y^\mathsf{T} \otimes \slashed{C}_y^\mathsf{T}$, and the operators $P_{\mathrm{A}}$ and $P_\mathrm{C}$ commute with all $A_x$ and $C_z$, respectively.
    
    \emph{Part~III.} Since Alice's and Charlie's observables satisfy the canonical anticommutation relations, by the Jordan--Wigner representation theorem there exist local unitaries such that
    \begin{equation}
        \mathcal{H}_{\mathrm{A}} \cong \bigotimes_{k=1}^r \mathbb{C}^{2}\otimes\mathcal{H}_{\mathrm{A}'},
        \qquad
        \mathcal{H}_{\mathrm{C}} \cong \bigotimes_{k=1}^r \mathbb{C}^{2}\otimes\mathcal{H}_{\mathrm{C}'},
    \end{equation}
    with $r=\lfloor m/2\rfloor$, and the observables take the canonical form
    \begin{align}
        \forall x,z\in\{0,1,\dots,2r\}:\quad A_x = \tilde A_x \otimes \mathds{1}_{\mathcal{H}_{\mathrm{A}'}},
        \qquad
        C_z = \tilde C_z \otimes \mathds{1}_{\mathcal{H}_{\mathrm{C}'}},
    \end{align}
    and for $m$ odd:
    \begin{equation}
        A_m=\tilde A_m \otimes M_0 +  \Phi(\tilde A_m) \otimes M_1, \quad C_m=\tilde C_m\otimes N_0 + \Phi(\tilde C_m) \otimes N_1
    \end{equation}
    for some $M_0+M_1=\mathds 1$ and $N_0+N_1=\mathds 1$ orthogonal measurements.

    From Part~II, $[P_\mathrm{A},A_x]=0$ and $[P_\mathrm{C},C_z]=0$ for all $x,z$. Since the first $2r$ observables $A_x$ and $C_z$ generate the full algebra $\mathcal B(\mathbb{C}^{2 \otimes r})\otimes \mathds{1}$, we obtain
    \begin{equation}
        P_\mathrm{A}=\mathds{1}_{2}^{\otimes r} \otimes P_{\mathrm{A}'},
        \qquad
        P_\mathrm{C}=\mathds{1}_{2}^{\otimes r} \otimes P_{\mathrm{C}'}.
    \end{equation}
    Consequently, each source state factorizes into a maximally entangled $2^r\times 2^r$ component and an ancillary component:
    \begin{equation}\label{eq:tensor_product_form}
        \ket{\psi}_{\mathrm{AB_1}} = \ket{\Phi_{2^r}}_{\mathrm{AB_1}}\otimes\ket{\xi}_\mathrm{A'B'_1},
        \qquad
        \ket{\psi}_{\mathrm{B_2C}} = \ket{\Phi_{2^r}}_{\mathrm{B_2C}}\otimes\ket{\xi}_{\mathrm{B'_2C'}}.
    \end{equation}
    Using Eq.~\eqref{eq:operators_bob}, Bob's observables act on the support of the state as
    \begin{align}
        B_y&=\sum_{x=1}^m\sum_{z=1}^m h_{xy}h_{zy}\,A_x^{\mathsf T}\otimes C_z^{\mathsf T}\nonumber\\
        &= \tilde B_y \otimes M_0^\mathsf{T}  N_0^\mathsf{T}+(\Phi\otimes \mathds 1)(\tilde B_y) \otimes M_1^\mathsf{T}  N_0^\mathsf{T}+ (\mathds 1\otimes\Phi)(\tilde B_y) \otimes M_0^\mathsf{T} N_1^\mathsf{T}+(\Phi\otimes \Phi)(\tilde B_y) \otimes M_1^\mathsf{T} N_1^\mathsf{T},
    \end{align}
    where $\tilde B_y$ is exactly the reference bilocal observable.
\end{proof}

\section{New family of bilocal ROCN Bell inequalities}\label{appendix:new_family}

Consider a family of $3 \times 4$ ROCN matrices generated from the matrix $h \vcentcolon= h_\mathrm{EBI}$, introduced in the main text, through transformations of the form $h \mapsto R h$, where $R$ is a $3 \times 3$ real matrix. Our aim is to find an ROCN matrix with the lowest possible bilocal bound $\beta^{Rh}_\mathrm{BL}$. However, first, we need to identify the structural conditions on the considered transformation required to preserve the ROCN properties. The orthogonality of rows implies
\begin{equation}
    R h h^\mathsf{T} R^\mathsf{T} = \mathrm{diag}(\lambda_1, \lambda_2, \lambda_3), 
\end{equation}
where $\lambda_1$, $\lambda_2$ and $\lambda_3$ are real numbers corresponding to the squared row lengths of the transformed matrix. Since $h h^\mathsf{T} = \mathds{1}_3$, this condition implies that $RR^\mathsf{T}$ must be diagonal. Next, normalization of the columns is equivalent to
\begin{equation}
    (h^\mathsf{T} R^\mathsf{T} R h)_{yy} = \sum_{x = 1}^{3} \sum_{z = 1}^{3} (h^\mathsf{T})_{yx} (R^\mathsf{T} R)_{xz} h_{zy} = 1,
\end{equation}
for all $y \in \{1,\ldots,4\}$. Denoting the $3 \times 3$ symmetric matrix $R^\mathsf{T} R$ by $S$, and using the explicit form of $h$, this condition can be rewritten as
\begin{equation}
    \sum_{x = 1}^{3} \sum_{z = 1}^{3} h_{xy} h_{zy} S_{xz} = \frac{1}{3} \Tr S + 2 \sum_{x<z}^{3} h_{xy} h_{zy} S_{xz} = 1.
\end{equation}
Summing over $y$ and using the ROCN properties of $h$, we obtain
\begin{equation}
    \frac{4}{3} \Tr S + 2 \sum_{x<z}^{3} \sum_{y=1}^4 h_{xy} h_{zy} S_{xz} = \frac{4}{3} \Tr S + 2 \sum_{x<z}^{3} \frac{4}{3} \delta_{xz} S_{xz} = \frac{4}{3} \Tr S = 4,
\end{equation}
which implies that $\Tr S = \lambda_1 + \lambda_2 + \lambda_3 = 3$. We are therefore left with a homogeneous system of equations for the off-diagonal entries of $S$:
\begin{equation}
    \sum_{x<z}^{3} h_{xy} h_{zy} S_{xz} = 0.
\end{equation}
Since the matrix $M$ defined by $M_{y,(x,z)} = h_{xy} h_{zy}$, with $y \in \{1,\ldots,4\}$ and $x,z\in\{1,\ldots,3\}$ such that $x<z$, has full column rank, it follows that all off-diagonal entries of $S$ must vanish. 

To sum up, a transformation $h \mapsto R h$ preserves the ROCN properties if the following conditions hold: (i) $RR^\mathsf{T}$ is diagonal, (ii) $R^\mathsf{T}R$ is diagonal, and (iii) the squared Frobenius norm of $R$ satisfies $\Tr RR^\mathsf{T} = \Tr R^\mathsf{T}R = 3$. The parametrization of $R$ then splits into three distinct cases, depending on the multiplicities of the squared row lengths $\lambda_1$, $\lambda_2$ and $\lambda_3$. For our purposes, however, an explicit parametrization is not needed. Numerical optimization over all matrices $R$ satisfying conditions (i)--(iii) shows that, up to permutations and reflections, the lowest bilocal bound, $\beta^\mathrm{min}_\mathrm{BL} = 2\sqrt{2}$, is achieved by the following ROCN matrix:
\begin{equation}
    h_\mathrm{min} = \frac{1}{\sqrt{2}} \begin{bmatrix}
        -1 & +1 & -1 & +1\\
        0 & -1 & 0 & +1\\
        -1 & 0 & +1 & 0
    \end{bmatrix}.
\end{equation}
However, this matrix has a serious drawback: it does not satisfy the self-testing condition introduced in Theorem~\ref{th:self_testing_statement}. To circumvent this problem, it is enough to apply a perturbation of the form $h_\mathrm{min} \mapsto R_z(\theta) h_\mathrm{min}$, where $\theta \in \left]0,\pi/2\right[$ and the $3 \times 3$ rotation matrix $R_z(\theta)$ is given by
\begin{equation}
    R_z(\theta) = \begin{bmatrix}
        1 & 0 & 0\\
        0 & \cos \theta & - \sin \theta\\
        0 & \sin \theta & \cos \theta
    \end{bmatrix}.
\end{equation}
The resulting one-parameter family of ROCN matrices $h_\theta \vcentcolon= R_z(\theta) h_\mathrm{min}$ takes the form
\begin{equation}
    h_\theta = \frac{1}{\sqrt{2}} \begin{bmatrix}
        -1 & +1 & -1 & +1\\
        +\sin \theta & -\cos \theta & -\sin \theta & +\cos \theta\\
        -\cos \theta & -\sin \theta & +\cos \theta & +\sin \theta
    \end{bmatrix}.
\end{equation}
In the limit $\theta \to 0$, the bilocal bound can be made arbitrarily close to $\beta_\mathrm{BL}^\mathrm{min}$. Furthermore, the corresponding matrix $M$ takes the form
\begin{equation}
    M_\theta = \frac{1}{2} \begin{bmatrix}
        -\sin \theta & +\cos \theta & -\sin \theta \cos \theta\\
        -\cos \theta & -\sin \theta & +\sin \theta \cos \theta\\
        +\sin \theta & -\cos \theta & -\sin \theta \cos \theta\\
        +\cos \theta & +\sin \theta & +\sin \theta \cos \theta
    \end{bmatrix}.
\end{equation}
The determinant of the $3\times 3$ submatrix formed by the first three rows is proportional to $\sin 2\theta$, which is nonzero for $\theta \in \left]0,\pi/2\right[$. Thus, $M_\theta$ has full column rank and the ROCN bilocal inequalities generated by the one-parameter family $h_\theta$ enable self-testing throughout this range.

\end{document}